\documentclass[conference,letterpaper]{IEEEtran}

\usepackage[utf8]{inputenc} 
\usepackage[T1]{fontenc}
\usepackage{url}
\usepackage{ifthen}
\usepackage{cite}
\usepackage{amssymb,amsfonts}
\usepackage{algorithmic}
\usepackage{graphicx}
\usepackage{textcomp}
\usepackage{amsthm,amssymb}
\usepackage[cmex10]{amsmath}
\usepackage{stmaryrd,txfonts}
\usepackage{mathrsfs}
\usepackage{algorithm}
\usepackage{threeparttable}
\usepackage{subfigure}
\theoremstyle{definition}
\newtheorem{prop}{Proposition}
\newtheorem{definition}{Definition}

\theoremstyle{remark}

\begin{document}
\title{High Throughput QC-LDPC Decoder With Optimized Schedule Policy in Layered Decoding} 


\author{%
  \IEEEauthorblockN{Dongxu Chang, Qingqing Peng, and Guanghui Wang}
  \IEEEauthorblockA{School of Mathematics \\
                    Shandong University\\
                    Jinan, China\\
                    Email: \{dongxuchang, pqing\}@mail.sdu.edu.cn,\\
                    ghwang@sdu.edu.cn}
                    
  \and
  \IEEEauthorblockN{Guiying Yan}
  \IEEEauthorblockA{Academy of Mathematics and Systems Science, CAS\\ 
                    University of Chinese Academy of Sciences\\
                    Beijing, China\\
                    Email: yangy@amss.ac.cn}
}

\maketitle


\begin{abstract}
   In this study, a scheduling policy of layered decoding for quasi-cycle (QC) low-density parity-check (LDPC) codes with high throughput and good performance is designed. The influence of scheduling on the delay of the decoder's hardware implementation and on the decoding performance are considered simultaneously. Specifically, we analyze the idle time required under various scheduling sequences within a pipelined decoding architecture and formulate the problem as a traveling salesman problem (TSP) aiming at minimizing idle time. Furthermore, considering that different scheduling sequences can affect decoding performance, we refine the graph used to solve the TSP based on scheduling characteristics that promote improved decoding outcomes. Simulation results demonstrate that the identified scheduling sequence achieves a low number of hardware delays while maintaining excellent decoding performance for 5G New Radio (NR) LDPC codes.
\end{abstract}


\section{Introduction}
\label{sec:introduction}
Low-density parity-check (LDPC) code \cite{gallager1962low} has been widely used in communication fields, such as 5G New Radio (NR) \cite{3gpp20185g}, IEEE 802.11 (WiFi) \cite{ieee2007ieee}. It was demonstrated by D. MacKay and M. Neal in 1996 \cite{mackay1997near} that LDPC codes can approach the Shannon limit. The belief propagation (BP) decoding algorithm \cite{mackay1999good} of LDPC code is most commonly used due to its low complexity and good performance in practice.

As a structured LDPC code, Quasi-cyclic (QC) LDPC codes feature a parity-check matrix constructed from circularly shifted identity submatrices \cite{fossorier2004quasicyclic}. Owing to their exceptional hardware compatibility and low encoding/decoding complexity, QC-LDPC codes have been adopted in the 5G New Radio (NR) standard for data channels \cite{3gpp20185g}.

In 2004, D. Hocevar introduced the concept of layered belief propagation (LBP) decoding\cite{hocevar2004reduced}, which updates information sequentially. This decoding strategy can accelerate the decoding process as it ensures that the most recent information is disseminated. In comparison to the flooding scheduling strategy \cite{kschischang1998iterative}, which updates all variable-to-check (V2C) messages or check-to-variable (C2V) messages simultaneously in one iteration, LBP can reduce decoding complexity by nearly 50$\%$ while achieving the same performance.

Through the pipeline decoding architecture \cite{bhatt2006pipelined} of LBP, QC LDPC can achieve high clock frequency and high data-throughput in hardware implementation. However, memory conflicts can arise when the information to be read from the current row has not yet been written in the previous row. In such cases, idle cycles are required to delay the reading of information, which causes memory conflicts. This increases the decoding latency and reduces the throughput of the decoder. It is noteworthy that the number of idle cycles required can vary depending on the update order of check nodes and variable nodes.

In \cite{sulek2010idle}, a random search algorithm was developed to reorder the rows and columns of the parity-check matrix, enabling the identification of decoding orders with fewer idle cycles. In \cite{petrovic2020flexible}, by employing a flooding scheduling strategy on certain layers, memory conflicts can be mitigated. In \cite{marchand2009conflict} and \cite{wu2015updating}, the required number of idle cycles was analyzed and characterized by the common degree of adjacent check nodes in the decoding order. By formulating the problem of finding the scheduling sequence with the minimum idle cycles as a Traveling Salesman Problem (TSP) \cite{davendra2010traveling}, scheduling sequences with low decoding latency can be obtained. Nevertheless, there is still room to further optimize the decoding order to reduce the number of required idle cycles. More importantly, existing studies that investigate the impact of decoding scheduling on hardware memory conflicts have not taken into account its effects on decoding performance. As demonstrated in \cite{chang2024optimization,tian2022novel,wang2020two}, different scheduling sequences can significantly impact both decoding efficiency and decoding performance.

In this study, similar to \cite{marchand2009conflict} and \cite{wu2015updating}, we transform the problem of finding the scheduling sequence that minimized the number of idle cycles into an equivalent TSP for a solution. Furthermore, based on the characteristics that an effective scheduling sequence should possess \cite{tian2022novel}, we modify the graph used for TSP accordingly, so that the scheduling sequence obtained from the TSP can achieve good decoding performance. We demonstrate that when the latency of the soft-output (SO) data path is small, the constraints imposed on the scheduling sequence from the perspective of decoding performance do not, on average, hinder our ability to find the scheduling with low idle time. Simulation results confirm that the scheduling sequences obtained by the proposed algorithm effectively resolve memory conflicts while maintaining excellent decoding performance.

The rest of the paper is organized as follows: Section \uppercase\expandafter{\romannumeral2} provides an introduction to the scheduling problem under LBP, as well as the calculation of the number of idle cycles required in the pipelined decoding architecture. Section \uppercase\expandafter{\romannumeral3} presents the policy for finding scheduling sequences that simultaneously achieve favorable idle time and decoding performance. Section \uppercase\expandafter{\romannumeral4} provides a comparative simulation analysis of various scheduling policies. We conclude the paper in Section \uppercase\expandafter{\romannumeral5}.

\section{Preliminaries}
\subsection{Belief Propagation for LDPC Codes}
A brief review of the iterative BP decoding algorithm \cite{mackay1999good} is provided to facilitate the analysis of scheduling sequences.

The message-passing process of BP decoding consists of V2C message updates and C2V message updates. The message update rule of V2C is 
\begin{equation}
	m_{i \rightarrow \alpha}^{(l)}=m_{0}+\sum_{h \in N(i) \backslash \alpha} m_{h \rightarrow i}^{(l-1)},
	\label{v2c}
\end{equation}
and the message update rule of C2V is 
\begin{equation}
	m_{\alpha \rightarrow i}^{(l)}=2 \tanh ^{-1}\left(\prod_{j \in N(\alpha) \backslash i} \tanh \left(m_{j \rightarrow \alpha}^{(l-1)} / 2\right)\right)
	\label{c2v}
\end{equation}
where $m_0$ is the channel message in LLR form, $l$ is the number of iterations, $N(v)$ represents the nodes connected directly to node $v$, $m_{i \rightarrow \alpha}^{(l)}$ means the message from variable node $i$ to check node $\alpha$ in iteration $l$ and $m_{\alpha \rightarrow i}^{(l)}$ means the message from check node $\alpha$ to variable node $i$ in iteration $l$. The initial message $m_{i \rightarrow \alpha}^{(0)}$ and $m_{\alpha \rightarrow i}^{(0)}$ is 0.

\subsection{Layered Belief Propagation and Scheduling Sequences}
In each decoding iteration, LBP \cite{hocevar2004reduced} sequentially selects each check node or layer in the graph and updates the information on all edges connected to the selected check node or layer, where one layer refers to all the rows obtained from a single row of the base graph after lifting. Specifically, it performs V2C message updates along all edges connected to the selected check node, followed by C2V message updates. The order in which check nodes are selected for updating in LBP is referred to as the scheduling sequence, defined as follows:

\begin{definition}[scheduling sequences]
	The scheduling sequence of a QC-LDPC code is a sequence composed of the indices of layers. It represents the order in which the layers are selected for decoding in LBP, with the decoder selecting layers for decoding according to the order in which they appear in the scheduling sequences.
\end{definition}

When the same message passing order is used in each iteration, the scheduling sequence can consist of only the layer update order within one iteration.

Different scheduling sequences can have a great impact on the message passing efficiency in layered decoding, thereby affecting the decoding performance and latency \cite{chang2024optimization}.

\subsection{Idle Cycles in Pipelining}
\begin{figure}
	\centering
	\includegraphics[width=9cm]{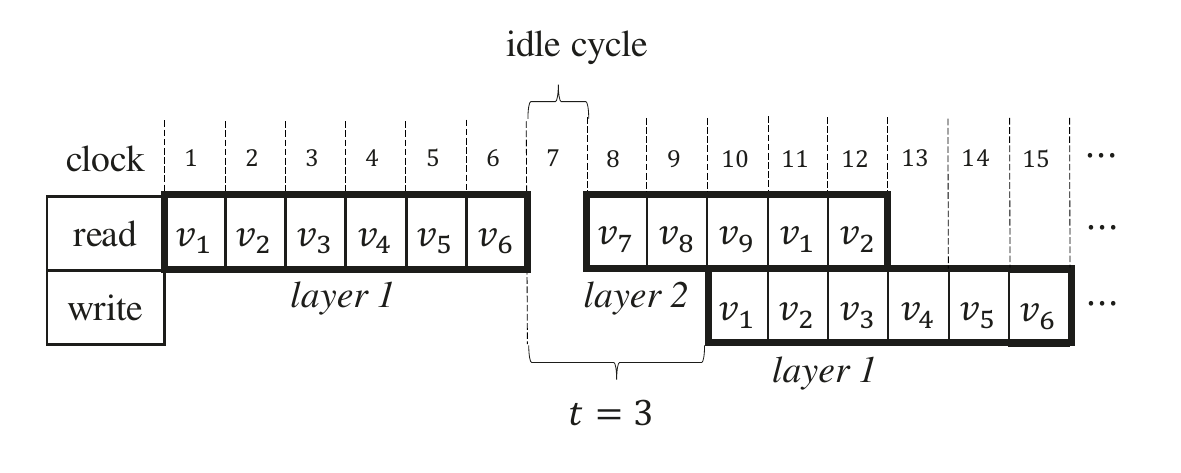}
	\caption{Idle cycle in the pipelined decoding architecture.}
	\label{fig_idletime}
\end{figure}
The pipelined decoding architecture achieves high throughput by effectively utilizing hardware resources \cite{bhatt2006pipelined, sun2006high}. As illustrated in Fig. \ref{fig_idletime}, the information associated with the variable nodes is read and written in a pipelined manner, where $t$ represents the latency of the SO data path, measured in clock cycles. During this process, the decoder simultaneously reads the information of the variable nodes in the current decoding layer while writing the updated information of the variable nodes from the previously decoded layer.

However, when there are correlations between the input and output information, memory conflicts can occur, necessitating the insertion of idle cycles to prevent such conflicts.
For example, in Fig. \ref{fig_idletime}, one idle cycle needs to be inserted to avoid memory conflicts. It is worth noting that in Fig. \ref{fig_idletime}, the read and write-back order of the information has been optimized to ensure the minimum number of required idle cycles \cite{rovini2007minimum,lin202133,nadal2021parallel}. Following \cite{rovini2007minimum} and \cite{ren2024generalized}, for a given scheduling sequence $\left\{c_1, c_2, \cdots, c_m\right\}$, the number of idle cycles $n_{idle}$ required in one iteration can be calculated using the following equation:
\begin{equation}
	\nonumber
	n_{idle} = max\Big(t-(d_{c_1}-d^{c}_{c_{m},c_1}),0\Big)+\sum_{i=2}^{m}max\Big(t-(d_{c_i}-d^{c}_{c_{i-1},c_i}),0\Big),
\end{equation}
where $d_{c_i}$ is the degree of layers $c_i$ and $d^{c}_{c_{i-1},c_i}$ is the common degree of layers $c_{i-1}$ and $c_i$.

\section{Scheduling with Small Idle Time And Good Performance}
To achieve high throughput and performance in scheduling, in this section, we transform the problem of finding the scheduling sequence with minimal idle time into an asymmetric TSP. By adjusting the graph used in the TSP, based on the characteristics of good scheduling sequences \cite{tian2022novel}, this approach simultaneously accounts for the impact of scheduling on both hardware design and decoding performance.

\subsection{The TSP for Solving Memory Conflicts}
The task of identifying the scheduling sequence that minimizes the number of idle cycles can be equivalently reformulated as an asymmetric TSP, defined as follows: for a given complete weighted direct graph $\mathcal{G}$, the objective is to find a Hamiltonian cycle with the minimum total weight \cite{davendra2010traveling}. In this context, the nodes in $\mathcal{G}$ represent each layer in the given QC-LDPC code. The weight $w_{i-1, i}$ of the edge from node $c_{i-1}$ to node $c_{i}$ in $\mathcal{G}$ represents the number of idle cycles required for decoding from layer $c_{i-1}$ to layer $c_{i}$, which is $w_{i-1, i} = max(t-(d_{c_i}-d^{c}_{c_{i-1},c_i}),0)$. It should be noted that the idle time from layer $c_{i-1}$ to $c_i$ may differ from the idle time from layer $c_i$ to $c_{i-1}$, meaning that the weights of the edges between two nodes in the $\mathcal{G}$ may not be the same in different direction. Figure \ref{fig_TSP} illustrates the graph $\mathcal{G}$ constructed for the asymmetric TSP with four layers. The resulting asymmetric TSP can be directly solved using the differential evolution algorithm, or it can be transformed into a symmetric TSP for subsequent resolution \cite{jonker1983transforming}.

\begin{figure}
	\centering
	\includegraphics[width=4cm]{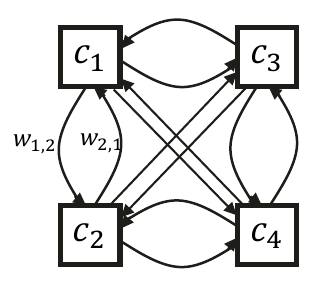}
	\caption{The graph $\mathcal{G}$ for TSP with 4 layers.}
	\label{fig_TSP}
\end{figure}

\subsection{Optimization of Idle Time And Decoding Performance}
The choice of scheduling sequences influences not only the number of required idle cycles but also the decoding performance. As demonstrated in \cite{tian2022novel}, scheduling sequences that achieve good decoding performance should prioritize updating check nodes with smaller degrees and fewer connections to punctured variable nodes.

Therefore, scheduling sequences with good decoding performance are expected to exhibit the following characteristics: 1. The layers in the scheduling sequence should be ordered by their degrees, from smallest to largest. 2. Layers with the same degree should be ordered by the number of connections to punctured variable nodes, from smallest to largest. Based on the above requirements, we group the layers in the graph such that the layers within each group have the same degree and the same number of connections to punctured variable nodes. The groups are further labeled based on $(d_i, p_i)$, in ascending order, where $d_i$ is the degree of layer $i$ and $p_i$ is the number of connections to punctured variable nodes of layer $i$. Specifically, group $i$ is labeled before group $j$ (i.e., $i < j$) if $d_i < d_j$, or $d_i = d_j$ and $p_i < p_j$. The labels of the layer groups are denoted sequentially from 1 to $P$.

\begin{figure}
	\centering
	\includegraphics[width=5.5cm]{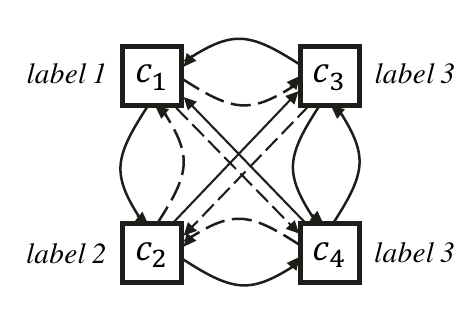}
	\caption{The modified graph $\mathcal{G}$ for TSP with 4 layers. The dashed edges represent weights of positive infinity.}
	\label{fig_newTSP}
\end{figure}

Consider an arbitrary edge in the directed graph $\mathcal{G}$. Let this edge point from a layer labeled $i$ to a layer labeled $j$. If \(i\) and \(j\) do not satisfy any of the following three conditions, set the weight of the corresponding edge to positive infinity: 1. $i+1=j$. 2. $i = j$. 3. $i = P, j = 1$. The third scenario above corresponds to the last updated check node in the current iteration transitioning to the first updated check node in the subsequent iteration. With these modifications, the resulting scheduling sequence can satisfy the characteristics of efficient scheduling. Additionally, for the edges directed from the layers labeled $P$ to the layers labeled $1$, we add an extra weight of $H$ to these edges, where $H$ is a large constant. This ensures that the resulting scheduling sequence strictly follows the order of updating the layers with smaller labels before those with larger labels within a single iteration. Figure \ref{fig_newTSP} illustrates the modified version of the graph in Figure \ref{fig_TSP}, where dashed edges represent edges with infinite weights.

By solving the asymmetric TSP problem on the modified graph $\mathcal{G}$, an optimal scheduling sequence with minimal idle time can be obtained within those that exhibit good decoding performance. Specifically, for the scheduling sequence $\left\{c_1, c_2, \cdots, c_m\right\}$, if $c_1$ and $c_m$ are carefully chosen such that the common degree from $c_m$ to $c_1$ is small, then $n_{idle}$ can be approximated by $n_{idle}^{\prime}$, where
\begin{equation}
	n_{idle}^{\prime} = \sum_{i=2}^{m}max\Big(t-(d_{c_i}-d^{c}_{c_{i-1},c_i}),0\Big).
\end{equation}
We show that if the latency of the SO data path $t$ is relatively small, then probabilistically, prioritizing the update of low-degree check nodes before high-degree check nodes can achieve the minimum $n_{idle}^{\prime}$.

\begin{prop}
	Let the variable nodes be uniformly randomly connected to the check nodes, and denote the minimum degree of the layers by $d_{min}$. Assume that $t\leq d_{min}$ and $N$ approaches infinity. Then prioritizing the update of low-degree check nodes before high-degree check nodes can yield the minimum expectation of $n_{idle}^{\prime}$.
	\label{prop1}
\end{prop}
\begin{proof}
	Appendix.
\end{proof}

\section{Simulation Results}
In this section, we compare the idle time and decoding performance of the scheduling sequences obtained under different scheduling policies. Specifically, we examine the scheduling policy derived from solving the TSP with the sole objective of minimizing idle cycles (denoted as "idle"), as well as the scheduling policy that simultaneously minimizes idle cycles and optimizes decoding performance (denoted as "idle\&performance"). Additionally, we include comparisons with the scheduling policy that orders layers based on their degree in ascending order (denoted as "LD") \cite{frenzel2019static} and the nested scheduling policy derived from threshold calculations (denoted as "nest") \cite{jang2022design}.

\begin{figure*}
	\centering
	\subfigure[]
	{
		\begin{minipage}[b]{0.45\linewidth}
			\includegraphics[width=7cm]{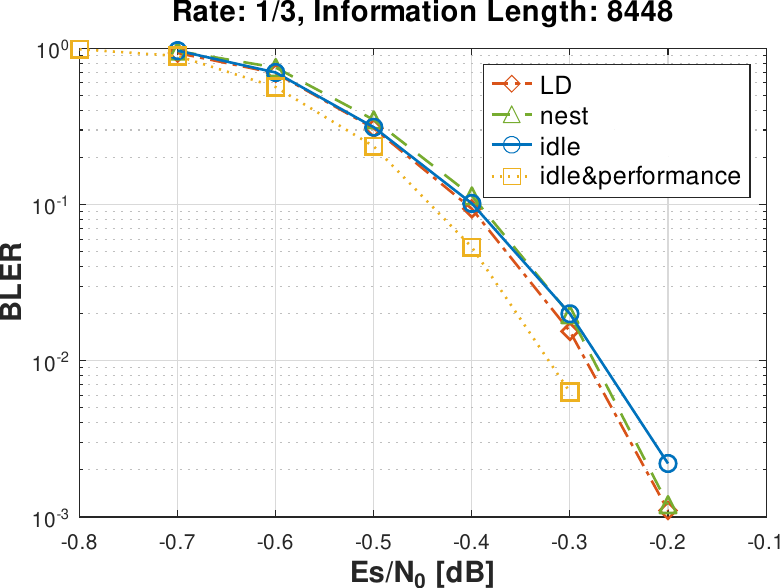}
		\end{minipage}
	}
	\subfigure[]
	{
		\begin{minipage}[b]{0.45\linewidth}
			\includegraphics[width=7cm]{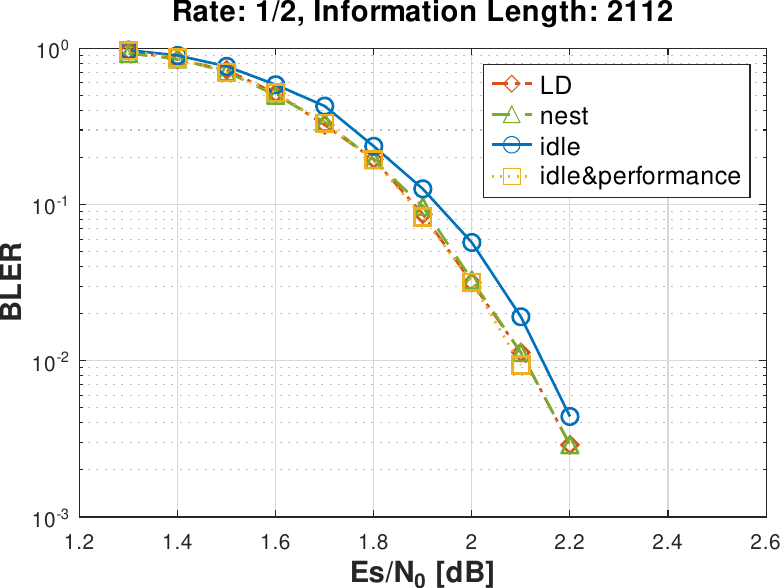}
		\end{minipage}
	}
	\caption{Error-correcting performance for various scheduling policies under different code rates and information lengths. The 5G NR LDPC codes constructed from BG1 are evaluated. The latency of the SO data path $t$ is 4.}
	\label{fig_simulation_t4}
\end{figure*}

\begin{figure*}
	\centering
	\subfigure[]
	{
		\begin{minipage}[b]{0.45\linewidth}
			\includegraphics[width=7cm]{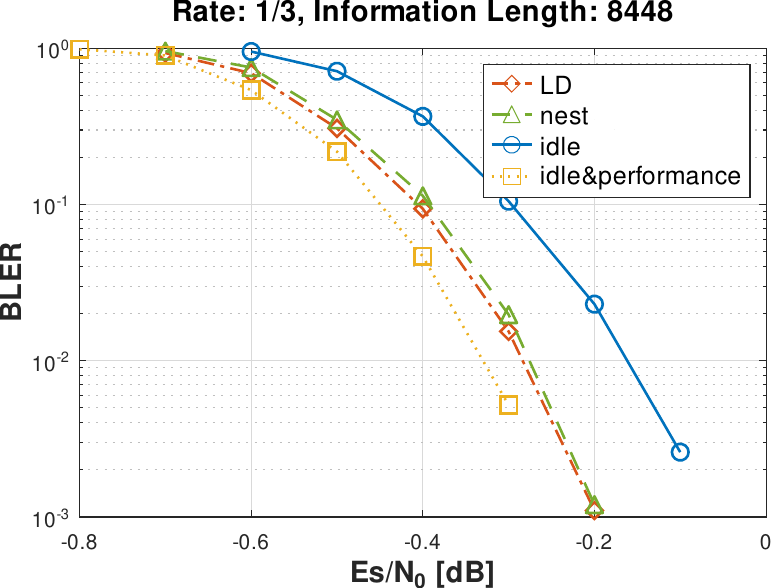}
		\end{minipage}
	}
	\subfigure[]
	{
		\begin{minipage}[b]{0.45\linewidth}
			\includegraphics[width=7cm]{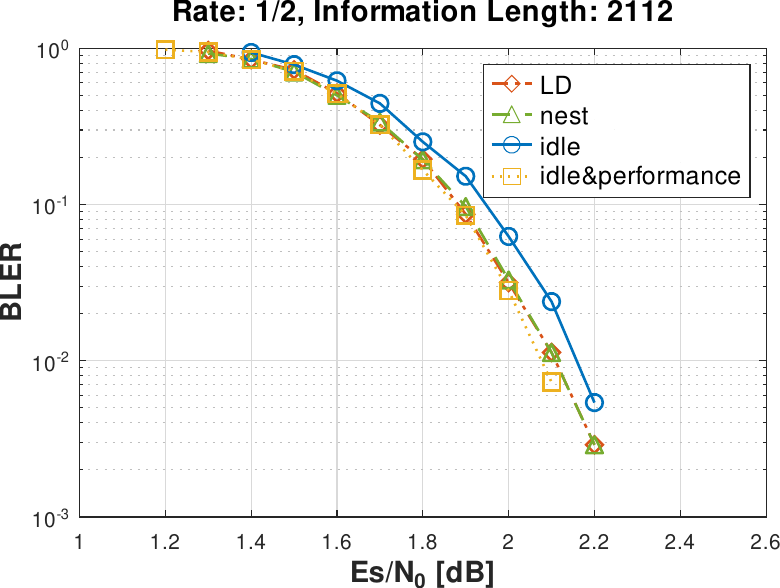}
		\end{minipage}
	}
	\caption{Error-correcting performance for various scheduling policies under different code rates and information lengths. The 5G NR LDPC codes constructed from BG1 are evaluated. The latency of the SO data path $t$ is 9.}
	\label{fig_simulation_t9}
\end{figure*}

\begin{table*}[]
	\centering
	\caption{$n_{idle}$ for different scheduling policies under 5G NR LDPC BG1.}
	\label{table_idletime}
	\setlength{\tabcolsep}{3pt}
	\begin{threeparttable}
		\begin{tabular}{|l|ll|ll|}
			\hline
			& \multicolumn{2}{c|}{$t=4$} & \multicolumn{2}{c|}{$t=9$} \\ \hline
			Scheduling policies   & \multicolumn{1}{l|}{$R=1/3, K=8448$}           & \multicolumn{1}{l|}{$R=1/2, K=2112$}  & \multicolumn{1}{l|}{$R=1/3, K=8448$}           & \multicolumn{1}{l|}{$R=1/2, K=2112$} \\ \hline
			nest \cite{jang2022design}                 & \multicolumn{1}{c|}{19}               & \multicolumn{1}{c|}{10} & \multicolumn{1}{c|}{214}               & \multicolumn{1}{c|}{99} \\ \hline
			LD \cite{frenzel2019static}                 & \multicolumn{1}{c|}{15}               & \multicolumn{1}{c|}{9} & \multicolumn{1}{c|}{195}               & \multicolumn{1}{c|}{89} \\ \hline
			idle                  & \multicolumn{1}{c|}{2}               & \multicolumn{1}{c|}{2} & \multicolumn{1}{c|}{158}               & \multicolumn{1}{c|}{66} \\ \hline
			idle\&performance                  & \multicolumn{1}{c|}{6}               & \multicolumn{1}{c|}{4} & \multicolumn{1}{c|}{176}               & \multicolumn{1}{c|}{79} \\ \hline
			\end{tabular}
		\end{threeparttable}
	\end{table*}
	
In Fig. \ref{fig_simulation_t4} and \ref{fig_simulation_t9}, the block error rate (BLER) performance for various scheduling policies is presented for \( t = 4 \) and \( t = 9 \), respectively, across different code rates and information lengths. The 5G NR LDPC codes constructed from base graph (BG) 1 are evaluated. It can be observed that, for different values of $t$, code rates, and information lengths, the scheduling sequences obtained solely by minimizing idle time exhibit large performance degradation compared to other scheduling policies. In contrast, scheduling methods that take both decoding performance and memory conflicts into account can mitigate the performance loss caused by focusing solely on memory conflict reduction. By characterizing the features of effective scheduling, optimal decoding performance can be achieved. 

In TABLE $\rm{\uppercase\expandafter{\romannumeral1}}$, the values of \(n_{\text{idle}}\) for different scheduling policies under 5G NR LDPC BG1 are presented. It can be seen that the number of idle cycles obtained when simultaneously considering idle time and decoding performance is slightly higher than that in the case where only idle time is considered. However, compared to other scheduling policies that solely focus on decoding performance, the number of idle cycles used under the proposed scheduling policy can be largely reduced.

\section{Conclusion}
In this study, we considered the impact of scheduling sequences on the decoding performance and hardware idle time of LDPC codes. By analyzing the number of idle cycles required for different scheduling sequences and the characteristics of effective scheduling sequences, we transformed the problem of minimizing idle cycles while maintaining good decoding performance into a TSP for resolution. Simulation shows that the scheduling sequences we obtained exhibit both good decoding performance and low idle time. Further research is required to better consider the characteristics of effective scheduling in order to enhance performance and reduce idle time.

\section*{Acknowledgment}
This work was partially supported by the National Key R\&D Program of China (2023YFA1009600).

\appendix
\section{Proof of Proposition \ref{prop1}}
\label{sec_appendix}
For two randomly chosen layers $c_i$ and $c_j$, let their degrees be denoted as $d_i$ and $d_j$, respectively. Since the variable nodes are uniformly randomly connected to the check nodes, the expected value $E_{i,j}$ of $max(t-(d_{c_j}-d^{c}_{c_{i},c_j}),0)$ is 
\begin{equation}
	\nonumber
	\begin{split}
		E_{i,j} &= \sum_{k=d_j+1-t}^{min(d_i,d_j)} P(d^{c}_{c_{i},c_j}=k)(k+t-d_j)\\
		&= \sum_{k=d_j+1-t}^{min(d_i,d_j)} \frac{C_{d_i}^{k}\cdot C_{N-d_i}^{d_j-k}}{C_{n}^{d_j}} (k+t-d_j) \\
		&= \sum_{k=d_j+1-t}^{min(d_i,d_j)} \frac{1}{N^{k}} + o(\frac{1}{N^{k}})\\
		& = \frac{1}{N^{d_j+1-t}} + o(\frac{1}{N^{d_j+1-t}}).
	\end{split}
\end{equation}

Therefore, the number of idle cycles required by layers with smaller degrees is expected to be significantly greater than that required by layers with larger degrees. As a result, for a layer with a small degree, the degree of the preceding layer during its update should be as small as possible to minimize the number of idle cycles required by the small degree layer. Since we assume $c_1$ and $c_m$ are carefully chosen such that the common degree from $c_m$ to $c_1$ is small, updating the layers in ascending order of their degrees can minimize the expectation of $n_{idle}^{\prime}$.

\bibliographystyle{ieeetr}
\bibliography{reference}

\end{document}